  \documentclass{article}

 \usepackage{mypaper}
\usepackage{hyperref}
\title{$\NP$-hard problems are not in $\BQP$}

\iflipics
\author{Reiner Czerwinski}{TU Berlin (Alumnus)}{}{}{}
\authorrunning{Reiner Czerwinski}
\Copyright{Reiner Czerwinski}
\keywords{NP versus BQP, BBBV theorem, P versus NP, halting problem, Complexity Classes}
\ccsdesc{Theory of computation~Complexity classes}
\ccsdesc{Theory of computation~Quantum complexity theory}
\relatedversiondetails{Previous Version}{https://arxiv.org/abs/2311.05624}
\else
\author{Reiner Czerwinski}
\fi

\newcommand{\unary}[1]{\ensuremath{1^{#1}}}
\begin{document}
 \begin{headdata}
  \begin{abstract}
    Grover's algorithm can solve \NP-complete problems on quantum
    computers faster
    than all the known algorithms on classical computers.
    However, Grover's algorithm still needs exponential time.
    Due to the BBBV theorem, Grover's algorithm
    is optimal for searches in the domain of a function, when
    the function is used as a black box.

    We analyze the \NP-complete set
\[\{ (\langle M \rangle, 1^n, 1^t ) \mid \text{ TM }M\text{ accepts an }x\in\{0,1\}^n\text{ within }t\text{ steps}\}.\]
If $t$ is large enough, then M accepts each word in $L(M)$ with length $n$ within $t$ steps.
So, one can use methods from computability theory to show that black box searching is the fastest way to find
a solution.
%
Therefore, Grover's algorithm is optimal for NP-complete problems.     
\end{abstract}
\end{headdata}
\newcommand{\infset}[1]{\ensuremath{D_{#1}}}
\newcommand{\boundset}[1]{\ensuremath{U_{#1}}}
\newcommand{\B}{\ensuremath{\{0,1\}}}
\section{Introduction}
\nocite{homeister2008quantum}
\nocite{arora2009computational}
One can efficiently simulate a classical computer with a quantum computer 
so that $\cP\subseteq\BQP$. However, there is no efficient quantum algorithm
known for
\NP-hard problems.

For any computable function $f :\{0,1\}^n \to \{0,1\}$ with
exactly one element $x$ such that $f(x)=1$, Grover's algorithm~\cite{grover1996fast} can find this element
$x$ in $\Theta(2^{n/2})$ accesses to $f$~\cite{shi2005quantum}.
There are several variants of Grover's algorithm for the case, that
the number of values $x$ with $f(x)=1$ is not exactly  one~\cite{ambainis2004quantum}.

Due to the BBBV theorem~\cite{bennett1997strengths}, Grover's algorithm is optimal for searching with a black box. A quantum computer
needs to apply at least $\Omega(2^{n/2})$ accesses to the black box.

In this paper, we construct an $\NP$-complete problem that we cannot
solve faster than with black-box searching.
We get an arbitrary TM $M$ and decide, whether there is an input word with
size $n$ that would be accepted within $t$ steps. If $t$ is great
enough, then any input word with size $n$ in $L(M)$ would be accepted
within $t$ steps. The number of steps $t$ would grow faster than any
computable function. So, we can use methods from computability theory
to prove, that we cannot be faster than in the black box manner.

This will infer that every $\NP$-hard problem is not in $\BQP$.

We take a look at the \NP-complete set
\[\{ (\langle M \rangle, 1^n, 1^t ) \mid \text{ TM }M\text{ accepts an }x\in\{0,1\}^n\text{ within }t\text{ steps}\}.\]
The TM $M$ will be fixed, so we denote
\[ \boundset{M}=\{\ (1^n, 1^t) \mid \text{ TM }M\text{ accepts an }x\in\B^n
  \text{ within }t\text{ steps} \}\text{.}
\] for an arbitrary but fixed $M$.

If a TM accepts an input, then the TM accepts it within a finite number of steps. So,
\[ L(M) = \lim_{t\to\infty} \{x \mid M \text{ accepts }x \text{ within }t\text{ steps} \} .\]

In section~\ref{se:inf}, we analyze the set 
\[ \infset{M} = \{ 1^n \mid \exists x \in \B^n \text{ with }x\in L(M) \}.\]
Obviously, for all $n\in\nat$:
\[ 1^n \in\infset{M} \iff \exists t\; (1^n, 1^t)\in\boundset{M}. \]

The set $\infset{M}$ is generally not computable.
But the set is computable relative to the oracle $L(M)$.
In this case, we would have to apply the oracle $L(M)$ on each $x\in \B^n$,
so we need a black box search.
We will conclude from the black box complexity of
$\infset{M}$ to the complexity of
the computable set $\boundset{M}$
in section~\ref{se:main}.
So, the set $\boundset{M}$ is in \NP{} and not computable faster
than with black box search.

\nocite{nielsen2001quantum}
\section{Notations and Preliminaries}

Let $M$ be a Turing machine. We declare $M(x)=1$ if $M$ accepts the input
$x$. Otherwise, $M(x)=0$. The language of a Turing machine
is the set of accepted words
\[ L(M) =\{ x \in \{0,1\}^* \mid M(x)=1 \}.
\]
If $M$ accepts $x$ within at most $t$ steps, then $M_t(x)=1$.
So, $M_t$ is a TM for each $t\in\nat$ and
\begin{equation}\label{lim}
  \lim_{t\to\infty} L(M_t) = L(M).
\end{equation}

If $n\in\nat$, then the unary encoding is defined as
 \[
  \unary{n} = \overbrace{1\dots1}^{n\text{ times}}.
 \]


In this paper, we use common abbreviations.
TM for Turing machine,
NTM for a nondeterministic Turing machine,
UTM for a universal Turing machine,
and OTM for an oracle Turing machine.

The class $\CE$ contains all c.e.\ sets, i.e.,
\[
\CE = \{ L \subseteq \B^* \mid \exists \text{ TM }M\text{ with }L=L(M) \}\text{.}
\]

In a decision problem, one will check if an input is in the set.
In a construction problem, one will find an element contained in the set.

The classes $\NP$ and $\CE$ are defined for decision problems.
But Grover's algorithm solves construction problems.
Fortunately, for an $\NP$-complete or $\CE$-complete set
the decision problem is as hard as the construction problem.
\begin{lemma}\label{decideconstruct}
  Let $L\subseteq\B^*$ be an $\NP$-complete or $\CE$-complete set.
  We assume, that we have an OTM with oracle for $L$. To find an input
  in $L$ with length $n$, one need not more than $n$ requests to the oracle $L$.
\end{lemma}
 \newcommand{\Cclass}{\ensuremath{\mathcal{C}}}
\begin{proof}
  Let $\Cclass\in\{\CE,\NP\}$. If a language $L$ is in the class $\Cclass$, then
  for an arbitrary but fixed $y_1 \dots y_k \in \B^k$ the set
  \[\{ x_1\dots x_k\overline{x} \in L \mid  x_1\dots x_k = y_1 \dots y_k \} \]
  is also in $\Cclass$.

  With the following algorithm one could find an $x\in L$:
  \begin{tabbing}
    bli\= bla   \= \kill
    for $k$ in $\{1,\dots,n\}$:\\
    \> $y_k := 1$\\
    \> if not ($\;\exists x\in\B^{n-k}$ with $y_1 \dots y_k x \in L\;$):\\
    \> \> $y_k:=0$\\
    $x=y_1\dots y_n$
  \end{tabbing}
  
\end{proof}
 \subsection{Rice's Theorem}

 \newcommand{\prop}{\ensuremath{\mathcal{P}}}
 If $\prop$ is a property of c.e.\ sets, then
 \[ \{ L\in\CE \mid \prop(L)\} \]
is an index set.

 Due to Rice's theorem\cite{RICE}, every non-trivial index set is
 uncomputable.
 An index set is trivial, if it is either empty or equal to $\CE$.
\begin{corollary}\label{byrice}
   If  $a,b\in\B^*$ with $a\not= b$,
   then the following index sets are undecidable: \\
   \newcommand{\itemx}{}
 \itemx $\{ L\in\CE\mid a\in L \land b\in L \}$,\\
   \itemx $\{ L\in\CE\mid a\in L \land b\not\in L \}$,\\
   \itemx $\{ L\in\CE\mid a\not\in L \land b\in L \}$,\\
   \itemx $\{ L\in\CE\mid a\not\in L \land b\not\in L \}$.   
\end{corollary}
\begin{proof}
   All these index sets are non-trivial.
\end{proof}

According to Rice's theorem, the set of all TM's with empty language is undecidable.
Therefore, there is no general algorithm to decide for each input word
whether it is in the language of a TM.
\begin{corollary}\label{faileverywhere}
  For every algorithm
  there is a TM $M$, such that the algorithm fails to decide
  for each word
  whether the word is in $L(M)$.
\end{corollary}
\begin{proof}
  Let us assume, there is an algorithm, that decides for a pair
  $(M,x)$ whether the TM $M$ accepts the word $x$.
  According to the second recursion theorem, a TM $M$ exists
  that applies the algorithm with itself and its input $x$.
  Then the TM $M$ rejects the input, if the algorithm predicts that
  $M$ will accept and vice versa.Thus, the algorithm cannot exist.
  (See~\cite{sipser1996introduction}).
\end{proof}
\subsection{Black Box Search}
A black box search problem is an optimization problem
where the internal working of the function is unknown or inaccessible.
In this paper, we use a Boolean function which is defined by an oracle access to the language
of an arbitrary TM. Therefore, the function is generally not computable.

We define the black box complexity 
as
the required number of function calls or oracle accesses.
A good survey of black box complexity for classical computers
can be found in \cite[Chapter 9]{wegener2005complexity} or in \cite{droste2006upper}.

If we use a quantum circuit instead of a TM for the black box search,
we denote it quantum black box search. According to Grover the running time
of quantum computers is asymptotically the square root of that of classical computers.
However, Grover's algorithm is optimal:
\begin{theorem}[BBBV]
  If $f$ is a function $f:\B^n\to B$ with exactly one solution for $f(x)=1$
  then we need $\theta(\sqrt(2^n))$ oracle applications to find the solution
  with quantum black box search.
\end{theorem}
The proof is given in \cite{bennett1997strengths}. If the function $f$
has $m$ solutions with $m\ge 1$, then
we need  $\theta(\sqrt(2^n/m))$ oracle applications~\cite[page 269-271, Chapter 6.6]{nielsen2001quantum}.
\section{c.e. Sets and Quantum Search}\label{se:inf}
We analyze the c.e.~set
\begin{equation}\label{eq:infset}
  \infset{M} =\{ 1^n \mid \exists x\in \B^n \text{ with } x\in L(M)\}
\end{equation}
A set is c.e.\ if it equals a language of a Turing machine.
Unfortunately, not every language of a TM is computable.
Therefore, there exists a TM $M$ where \infset{M} is not computable.
But \infset{M} is always 
computable relative to the oracle $L(M)$.
In this case, we can use a Grover-like algorithm for black-box search,
but they are optimal.
\begin{theorem}\label{blackbox}
  There exists a TM $M$,
  for which finding a word of length $n$ in $L(M)$ by using
  an OTM or a quantum computer with an oracle for $L(M)$
  requires a black box search.
  To find a word with length $n$ in $L(M)$, this mashine needs a black box search.
\end{theorem} 
\begin{proof}
    By Corollary~\ref{faileverywhere}, there exists a TM $M$
  where we need an oracle application
  for each word. We need an oracle access for testing
  whether the word is in $L(M)$, even if we know it for
  other words by  Corollary~\ref{byrice}.
To find a word with length $n$ in $L(M)$, we need a black box search. 


  
\end{proof}
  Due to Lemma~\ref{decideconstruct},
  black box search
  is also a lower bound for the decision problem.
  Thus, a quantum computer needs $\Omega(2^{n/2})$ oracle applications
  to decide whether $1^n\in\infset{M}$,
  in the worst case.


\section{\NP{} versus \BQP}\label{se:main}
The set $\infset{M}$ defined in (\ref{eq:infset}) is not computable.
Recall the definition of $\boundset{M}$ in the introduction:
\begin{equation}\label{eq:boundset}
  \boundset{M} =\{ (1^n,1^t) \mid \exists x\in \B^n \text{ with } x\in L(M_t)\}
\end{equation}
The set $\boundset{M}$ is similar to $\infset{M}$, except it is computable.
The set $\B^n$ is finite. So, equation (\ref{lim}) implies
\begin{equation}\label{eq:eqsets}
 1^n \in\infset{M} \iff \exists t\in\nat : (1^n,1^t) \in \boundset{M}
\end{equation}

\begin{lemma}
For any TM $M$ the set $\boundset{M}$ is in \NP.
\end{lemma}
\begin{proof}
  One wants to check whether $(1^n, 1^t)\in\boundset{M}$.
  
  An NTM can choose an $x\in\B^n$.
  After that, it can simulate the calculation of $M$ for $t$ steps.
  The run time of the NTM is $n+t$, which equals the length
  of the element $(1^n, 1^t)$.
\end{proof}

For the set $\infset{M}$ defined in (\ref{eq:infset}) we have proven
that a computer, even a quantum computer, cannot search faster than
in the black box manner
even with an oracle for $L(M)$. Now we will conclude this result
to the set $\boundset{M}$.

\begin{theorem}\label{limitblackbox}
  For an algorithm to test
  whether
  $(1^n, 1^t) \in \boundset{M}$ for any TM $M$,
  there exists a TM $M$, such that for every $n\in\nat$ there exists a $t\in\nat$,
  where testing whether
  $(1^n, 1^t) \in \boundset{M}$
  cannot be done
  faster than with a black box search.
\end{theorem}
\begin{proof}
  For a fixed $n$, there is a $t\in\nat$, such that
  $ 1^n \in\infset{M} \iff (1^n,1^t) \in \boundset{M}$.
  See equation (\ref{eq:eqsets}).
  In this case for all $x\in\B^n$: $x\in L(M_t) \iff x\in L(M)$. 

  Instead of the oracle $L(M)$, one can use the computable oracle $L(M_t)$.
  For a sufficiently large $t$, it does not matter if one uses $L(M)$ or $L(M_t)$.
  So, testing whether $(1^n, 1^t) \in \boundset{M}$ with oracle $L(M_t)$
  is as fast as testing whether $1^n\in\infset{M}$ 
  with oracle $L(M)$.
  Due to theorem~\ref{blackbox}, one cannot test it faster than with
  the black box search.

  Obviously, a TM or quantum computer without an oracle for $L(M_t)$ is not faster than a machine with this oracle.
\end{proof}
According to Theorem~\ref{limitblackbox}, there exists a TM $M$
where we need $\Omega(\sqrt(2^n))$ steps to test
whether
$(1^n, 1^t) \in \boundset{M}$.
The padding length $t$ can grow very fast. Thus, we have to show the
lower bound for run time for small $t$ in the next theorem.
\begin{theorem}\label{limitwithsmallt}
  For an algorithm to test
  whether
  $(1^n, 1^t) \in \boundset{M}$ for any TM $M$,
  there exists a TM $M$ where
  for each $t,n\in\nat$ we need black box search.
\end{theorem}
\begin{proof}
  Maybe there is an $x\in\B^n$ with $x\in L(M)$, but $x\not\in L(M_t)$.
  Assume, that one could test whether $(1^n, 1^t) \in \boundset{M}$
  faster than with black box search in this case. Then one
  could decide if there is such an $x$. But this is undecidable.
  So, one needs a black box search for any $t$ when the TM $M$ is arbitrary.
\end{proof}

Theorem 2 and the BBBV theorem imply that a quantum computer has the run time of $\Omega(2^{n/2})$ to decide whether $(1^n,1^t)\in \boundset{M}$ for $\boundset{M}$ defined in (\ref{eq:boundset}) in the worst case.
This implies the main results:
\begin{corollary}\label{NPBQP}
$\NP\not\subseteq\BQP$
\end{corollary}
\begin{proof}
  A TM $M$ exists, where the set $\boundset{M}$ described in (\ref{eq:boundset}) is in $\NP$
  but not in $\BQP$.
\end{proof}
\begin{corollary}
$\BQP$ does not include \NP-hard problems.
\end{corollary}
\begin{proof}
  Let $X$ be an \NP-hard set. Assume that $X\in\BQP$.
  Then there is a reduction $\boundset{M}\le^P_T X$ for the set $\boundset{M}$ described in (\ref{eq:boundset})
  with an arbitrary TM $M$.
  So, $\boundset{M}\in\BQP$ for any TM $M$, which is refuted by theorem~\ref{blackbox}.
\end{proof}
\begin{corollary}
$\cP\not=\NP$
\end{corollary}
\begin{proof}
  Assume $\cP=\NP$. Thus,
  $\NP\subseteq\BQP$ because of $\cP\subseteq\BQP$.
  But $\NP\subseteq\BQP$ is refuted by corollary~\ref{NPBQP}.
\end{proof}
\section{Conclusion}
The result of this paper includes the solution to the famous
$\cP$ vs.\! $\NP$ problem.
As Oded Goldreich mentioned on his web page~\cite{odedPvsNP},
one needs a novel insight to solve this problem.
In this paper we use padding arguments. These are often
used in complexity theory. But in all examples in books
on computational complexity~\cite{arora2009computational}, the padding length
is bounded by a computable function. The padding length in our
proof can grow faster than any computable function.
Therefore, we can use methodes from computability theory.

A proof for $\cP$ vs. $\NP$ has to
circumvent the relativization barrier~\cite{baker1975relativizations}.
Why does this proof method not relativize?
The proof in this paper uses padding
with only polynomially-long oracle queries,
which generally do not relativize~\cite{aaronson2008new}.

\section*{Acknowledgments}
I acknowledge
Charles H Bennett, Ethan Bernstein, Gilles Brassard, and Umesh Vazirani
for the remarkable BBBV theorem. This theorem shortened this paper.

Thanks to Michael Chavrimootoo~\cite{chavrimootoo2024}
for finding some issues in the second version of this paper on arxiv~\cite{czerwinski2024nphardproblemsbqp}.
\bibliography{lit,quantum}{}
\bibliographystyle{plain}
\end{document}